\documentclass[5p,times]{elsarticle}
\newif\ifProof
\Prooftrue 

\usepackage[T1]{fontenc}
\usepackage[latin1]{inputenc}
\usepackage[english]{babel}

\usepackage{fancybox,shadow,xcolor} 
\usepackage{setspace}
\usepackage{subfig}
\usepackage{placeins}
\usepackage{interval}
\usepackage{xcolor}
\usepackage{psfrag}
\usepackage{bm}
\usepackage{float}

\usepackage{graphicx}
\usepackage{amssymb,amsmath,amsthm}
\usepackage{mathrsfs}
\usepackage{mathdots}
\usepackage{algorithmic,algorithm}

\newtheorem{thm}{Theorem}
\newtheorem{lem}{Lemma}
\newtheorem{cor}{Corollary}
\newtheorem{rem}{Remark}
\newtheorem{prop}{Proposition}

\theoremstyle{definition}
\newtheorem{definition}{Definition}
\newtheorem{assumption}{Assumption}

\usepackage{interval}

\DeclareMathOperator*{\argmin}{arg\,min}

\DeclareMathOperator*{\diag}{diag}

\DeclareMathOperator*{\tr}{tr}

\renewcommand{\Re}{\mathbb{R}}

\usepackage{bbm}
\renewcommand{\paragraph}[1]{\smallskip\noindent\textbf{#1.} }

\newcommand{\BM}{\begin{bmatrix}}
\newcommand{\EM}{\end{bmatrix}}

\newcommand{\BBM}{\big[\begin{matrix}}
\newcommand{\EEM}{\end{matrix}\big]}

\newcommand{\bbm}{[\begin{matrix}}
\newcommand{\eem}{\end{matrix}]}

\newcommand{\rev}[1]{\textcolor{black}{#1}}

\begin{document}

\begin{frontmatter}   

\title{An interval-valued recursive estimation framework for  linearly parameterized systems}

\author[ECL]{Laurent Bako, Seydi Ndiaye, Eric Blanco} 
\ead{laurent.bako@ec-lyon.fr,eric.blanco@ec-lyon.fr}

\address[ECL]{Univ Lyon, Ecole Centrale Lyon, INSA
Lyon, Universit\'{e} Claude Bernard Lyon 1,
CNRS, Amp\`{e}re, UMR 5005, 69130 Ecully,
France}

\begin{abstract}
This paper proposes a recursive interval-valued estimation framework for identifying the parameters of  linearly parameterized systems which may be slowly time-varying. It is assumed that the model error (which may consist in  measurement noise or model mismatch or both) is unknown but lies at each time instant in a known interval. 
In this context, the proposed method relies on bounding the error generated by a given reference point-valued recursive estimator, for example, the well-known recursive least squares algorithm. We discuss the trade-off between computational complexity and tightness of the estimated parametric interval. 
\end{abstract}


\end{frontmatter}

\section{Introduction}
In system identification, the estimation problem refers to the task of finding the parameters of a given parametrized model family in such a way that the resulting model matches (in some sense) a set of data. The main challenge with this task is most presumably how to deal with the uncertainty affecting the data with regards to the assumed model structure (e.g., in the form of model error or  measurement noise). To hope for good estimates when the uncertainty is not negligible, it is important to model somehow the uncertainty.   Probabilistic distributions are probably the most common models for describing uncertainties in many engineering fields. Such a modelling, when accurate, can lead to the design of better estimation schemes. A problem however is that a fine probabilistic modeling of the uncertainty may require a strong prior knowledge of the process being modelled. While such a reliable knowledge is rarely available, probabilistic models of the uncertainty may be severely wrong hence damaging the performance of the estimator.  
An alternative approach to the probabilistic one is to assume that the uncertain variables (e.g., the noise component) of the model, although unknown, live in bounded and predefined sets. This corresponds to the so-called set-membership representation of the uncertainty. In this latter setting the underlying idea of the parameter estimator design is to characterize the entire set of parameters which, through the induced models, are consistent with the data samples and the uncertainty sets.  The literature of system identification  abounds in such set-memberships approaches, see e.g.,  \cite{Lauricella20-TAC,Milanese11-Automatica,Jaulin01-Book,Milanese13-Book,Chisci98-Automatica,Vicino96-TAC,Walter92-IFAC}. Various predefined geometrical forms  can be  considered for representing the  parameter sets but we restrict our attention here to the case of intervals \cite{Jaulin01-Book,Moore09-Book}. Assuming that the model error (uncertainty) takes values in a known sequence of intervals, one can estimate intervals containing all the parameters which are consistent with the data, the model and its uncertainty. Among the existing methods which have tackled this question we can cite \cite{Kieffer05,vanKampen11} for batch mode of estimation and  \cite{Sun03-IET,Gutman94} for recursive (online) mode of learning.

In this paper, we consider the problem of deriving a recursive interval-valued estimator for linearly parameterized models subject to  a bounded uncertainty. The data model is assumed to be linear with respect to the parameters (although the input-output map may indeed be nonlinear). Then, under the assumption that the model error sequence is only known to lie in some interval bounds, we first construct a tight interval-valued estimator based on the error generated by the recursive least squares (RLS) algorithm \cite{Goodwin09-Book,Tao03-Book,Bruggemann21-Automatica}. However, this (tight)  interval-valued estimator suffers from a level of computational complexity which is not affordable in practice when the estimation horizon grows towards infinity. We therefore turn to a family  of approximate implementations whose complexity can be calibrated in function of the desired level of performance (measured here in term of tightness of the interval-valued estimate). That is, the proposed family of estimators offer the user the possibility to tie the size of the desired  interval-valued estimate to the available computational resources. The proposed estimation framework applies to both stationary systems  and slowly time-varying ones. \\
Closely related works to the current paper are the ones reported in \cite{Sun03-IET} and \cite{Gutman94}. The former solves at each time a linear program on a sliding window of constant length. The latter bounds the error generated by a bank of  RLS identifiers. The current paper proposes a complementary development of this latter idea by considering a more general framework with regards to the model error representation. Moreover, our design method appears to be more systematic as it leads to a family of parametrizable interval-valued estimators. Also, it applies to both time-invariant and time-varying models with bounded change rate. In particular, it is shown in the time-invariant case that by applying an appropriate intersection operation,  the size of the estimate is guaranteed to decrease monotonically.

\paragraph{Outline} \rev{Section \ref{sec:Problem} states the problem of recursive interval-valued estimation and outlines some necessary preliminaries on interval set-membership representation of uncertainty. Our method for designing interval-valued estimators relies on the error generated by a reference point-valued recursive estimator. Hence, Section \ref{sec:ref-adaptive-identifier} discusses one possible candidate for such a reference point-valued estimator, the recursive least squares (RLS) estimator with exponential forgetting factor. In Section \ref{sec:interval-estimator}, we present the main estimator and its variants.  In Section \ref{sec:LTV}, we consider the more general estimation setting where the to-be-estimated model is no longer constant but slowly time-varying. Section \ref{sec:simulations} reports some simulation results to illustrate and analyze the performance of the proposed estimation framework. Finally,  Section \ref{sec:conclusion} presents some conclusions. 
}

\paragraph{Notation}
$\Re$ (resp. $\Re_+$) is the set of real (resp. nonnegative real) numbers; $\mathbb{Z}$ (resp. $\mathbb{Z}_+$) is the set of  (resp. nonnegative) integers. For  a real  number  $x$, $|x|$ will refer to the absolute value  of $x$.  For $x=\bbm x_1 & \cdots & x_n\eem^\top\in \Re^n$, $\left\|x\right\|_p$ will denote the $p$-norm of $x$ defined by $\left\|x\right\|_p=(|x_1|^p+\cdots+|x_n|^p)^{1/p}$, for $p\geq 1$. In particular for $p=\infty$, $\left\|x\right\|_\infty=\max_{i=1,\ldots,n}\left|x_i\right|$. For a matrix $A\in \Re^{n\times m}$, $\left\|A\right\|_F$ is the Frobenius norm of $A$ defined by  $\left\|A\right\|_F=\tr(A^\top A)^{1/2}$ (with $\tr$ referring to the trace of a matrix). 

 If $A=[a_{ij}]$ and $B=[b_{ij}]$ are real matrices of the same dimensions, the notation $A\leq B$ will be understood as an elementwise inequality on the entries, i.e., $a_{ij}\leq b_{ij}$ for all $(i,j)$.  $|A|$ corresponds to the matrix $[|a_{ij}|]$ obtained by taking the absolute value of each entry of $A$.  
In case $A$ and $B$ are real square  matrices, $A\succeq B$ (resp. $A\succ B$) means that $A-B$ is positive semi-definite (resp. positive definite). $I_n$ will denote the identity matrix of dimension $n$. 

\section{Problem statement}\label{sec:Problem}
We consider a linearly parameterized\footnote{Note that the considered system may indeed be nonlinear in term of input-output relation. For example, $x(t)$ may be of the form $x(t)=\phi(z(t))$ with $\phi$ being a known nonlinear map and $z(t)$ is formed from measurements.} discrete-time dynamic system defined by 
\begin{equation}\label{eq:system}
	y(t) = x(t)^\top \theta^\circ+v(t), 
\end{equation}                                           
where $y(t)\in \Re$ is the measured output at the discrete time $t\in \mathbb{Z}_+$, $x(t)\in \Re^n$ is the (known) regressor and $v(t)\in \Re$ denotes an (unknown)  noise component or a modeling error. 
The regressor $x(t)$ may, among other possibilities, assume a structure of the form 
$$
\begin{aligned}
	x(t)=	\left[y(t-1) \:   \cdots \:  y(t-n_a) \:  u(t)^\top \:  u(t-1)^\top \:  \cdots \: u(t-n_b)^\top \right]^\top 
\end{aligned}
$$
where $u(t)\in \Re^{n_u}$ is the input of the system and the integers $n_a$ and $n_b$ are the model orders. 
$\theta^\circ\in \Re^n$ in \eqref{eq:system} is an unknown constant parameter vector which is  to be estimated from data.  
The problem we consider in this paper is the following: given data points $\left\{\left(y(k),x(k)\right)\right\}_{k=1}^t$ generated by the system \eqref{eq:system} up to time $t$, we want to infer an estimate of the  parameter vector $\theta^\circ$. However, since the sequence $\left\{v(t)\right\}$ is unknown here we can hardly hope for an exact recovery of $\theta^\circ$. Hence we consider the scenario where $v(t)$ is componentwise bounded for all $t\in \mathbb{Z}_+$ with known bounds and setup as our objective to characterize a set-valued estimate which is guaranteed to contain $\theta^\circ$ while being consistent with the observed data.   

\begin{assumption}\label{assum:Bounding}
There exist (known) bounded sequences $\left\{(\underline{v}(t), \overline{v}(t))\right\}$ such that the noise sequence $\left\{v(t)\right\}$ in \eqref{eq:system} satisfies 
 $\underline{v}(t)\leq v(t)\leq  \overline{v}(t) $ for all $t\in\mathbb{Z}_+$.
\end{assumption}

\subsection{Some preliminaries on interval representation}
Consider two vectors $\underline{x}$ and $\overline{x}$  in $\Re^n$ such that $\underline{x}\leq \overline{x}$ with the inequality holding componentwise.  An interval $\interval{\underline{x}}{\overline{x}}$ of $\Re^n$ is the subset of $\Re^n$ defined by
\begin{equation}\label{eq:Interval}
	\interval{\underline{x}}{\overline{x}}= \big\{x\in \Re^n: \underline{x}\leq x \leq \overline{x}\big\}. 
\end{equation} 
An interval $\interval{\underline{x}}{\overline{x}}$ of  $\Re^n$ can be equivalently represented by
\begin{equation}
	\mathscr{I}(c_x,r_x)\triangleq \big\{c_x+\diag\big(r_x\big) \alpha : \alpha\in \Re^n, \: \left\|\alpha\right\|_\infty \leq 1\big\}
\end{equation}
where 
 \begin{equation}
	 c_x=\dfrac{\overline{x}+\underline{x}}{2},  \quad r_x = \dfrac{\overline{x}-\underline{x}}{2}
 \end{equation}
The notation $\diag(v)$ for a vector $v\in \Re^n$ refers to  the diagonal matrix whose diagonal elements are the entries of $v$.   We will call the so-defined $c_x$ the \textit{center or mid-point} of the interval $\interval{\underline{x}}{\overline{x}}$ and $r_x$ its \textit{radius} (a half of the width).
 To sum up, the interval set can be equivalently represented by the pairs $(\underline{x},\overline{x})\in \Re^n\times \Re^n$ and $(c_x,r_x)\in \Re^n\times \Re_+^n$  so that $\interval{\underline{x}}{\overline{x}}=\mathscr{I}(c_x,r_x)$.   Finally, it will be useful to keep in mind for the rest of the paper that $\underline{x}=c_x-r_x$ and $\overline{x}=c_x+r_x$. 

\begin{definition}[Parametric interval estimator]\label{def:Interval-Estimator}
Consider the system \eqref{eq:system} under Assumption \ref{assum:Bounding} and let $V^t=\big((\underline{v}(0),\overline{v}(0)), \ldots,(\underline{v}(t),\overline{v}(t))\big)$ and $Y^t=\big(y(1), \ldots,y(t)\big)$. Consider a dynamical system defined by 
\begin{equation}\label{eq:interval-estimator}
	\begin{aligned}
		&\underline{\theta}(t)=F_t\big(V^t,Y^t,\underline{\theta}(0),\overline{\theta}(0)\big) \\
		&\overline{\theta}(t)=G_t\big(V^t,Y^t,\underline{\theta}(0),\overline{\theta}(0)\big)
	\end{aligned} 
\end{equation}
where $F_t$ and $G_t$ are some functions indexed by time,  $(\underline{\theta}(t), \overline{\theta}(t))$ denote the output (or the state) of the system for any $t\in \mathbb{Z}_+$. 
The system \eqref{eq:interval-estimator} is called a (parametric) \textit{interval-valued estimator} for the parameter vector $\theta^\circ$ of system \eqref{eq:system} if: 
\begin{enumerate}
	\item[\textbf{(a)}] $\underline{\theta}(t)\leq \theta^\circ\leq \overline{\theta}(t)$  for all $t\in \mathbb{Z}_+$, whenever $\underline{\theta}(0)\leq \theta^\circ\leq \overline{\theta}(0)$ 
	\item[\textbf{(b)}] \eqref{eq:interval-estimator} is Bounded Input-Bounded Output (BIBO) stable i.e., if the signals $v$ and $y$ and the initial state $(\underline{\theta}(0),\overline{\theta}(0))$ are all bounded then so is $(\underline{\theta}, \overline{\theta})$. 
\end{enumerate}
\end{definition}
\noindent Now we recall from \cite{Bako18-TR,Bako19-Automatica} a lemma that will play a central role in the design of interval-valued estimators. 
\begin{lem}\label{lem:Az+w}
Let $M\in \Re^{n\times m}$ and $(\underline{z},\overline{z})\in \Re^m\times \Re^m$ such that  $\underline{z}\leq \overline{z}$.  Consider the set $\mathcal{I}$  defined by $\mathcal{I} = \big\{Mz: \underline{z}\leq z \leq \overline{z} \big\}$.  
Define the vectors $(c,r)$ by
 \begin{equation}\label{eq:C(c,p)}
	 \begin{aligned}
		 &c = Mc_z \\
		 &r = \left|M\right|r_z, 
	\end{aligned}
\end{equation}
 with $c_z=(\overline{z}+\underline{z})/2$ and $r_z=(\overline{z}-\underline{z})/2$.\\
Then $\interval{c-r}{c+r}$  is the tightest interval containing $\mathcal{I}$.
\end{lem}

\noindent We now state formally the estimation problem. \\
\paragraph{Problem}
Given the data $\left(y(k),x(k)\right)_{1\leq k\leq t}$ generated by system \eqref{eq:system} up to an arbitrary time $t\in \mathbb{Z}_+$, the uncertainty bounds $\left\{(\underline{v}(t), \overline{v}(t))\right\}_{1\leq k\leq t}$ on the noise sequence  as defined in Assumption \ref{assum:Bounding} and a prior (initial) interval set  $\mathscr{I}(c_\theta(0),r_\theta(0))$ containing the true parameter vector $\theta^\circ$ from \eqref{eq:system},  we are interested in finding an interval-valued estimate of the form $\mathscr{I}(c_\theta(t),r_\theta(t))\subset \Re^n$ (in the sense of Definition \ref{def:Interval-Estimator}), of the parameter vector $\theta^\circ$ in \eqref{eq:system} which is consistent with data. Moreover, it is desirable that the estimate  $(c_\theta(t),r_\theta(t))$ at time $t$ be obtained by a simple update mechanism from the measurements $(x(t),y(t),\underline{v}(t),\overline{v}(t))$ at time $t$ and a finite number $m$ of past estimates $(c_\theta(t-i),r_\theta(t-i))$, $i=1,\ldots,m$.

\noindent \rev{We will  describe in Section \ref{sec:interval-estimator} a framework for deriving a solution to this problem.} \\
\rev{Our method for constructing a recursive set-valued estimator requires three ingredients: (a) a reference adaptive point-valued identifier; (b) a characterization of the stability of the associated error dynamics ; (c) an appropriate mechanism for deducing the set-valued estimate from the point-valued one. Many recursive identifiers may be suitable for the role (a) mentioned above. Here however we choose to discuss only the RLS algorithm. }

\section{ A  reference adaptive identifier}\label{sec:ref-adaptive-identifier}
For the purpose of designing the recursive interval-valued estimator as stated above, we first study a reference adaptive point-valued identifier. 
\subsection{Recursive least squares (RLS)}
A candidate adaptive identifier for point (a) above is  the exponentially weighted recursive least squares (RLS) algorithm which returns a point-valued estimate $\theta(t)$ of $\theta^\circ$, selected at each time $t$ to be the minimizing point of an objective function $\theta\mapsto V_t(\theta)$,  
\begin{equation}\label{eq:optimization-RLS}
	\theta(t)=\argmin_{\theta\in \Re^n} V_t(\theta),
\end{equation}
with $V_t(\theta)$  defined by 
\begin{equation}\label{eq:Vt}
	V_t(\theta)=\dfrac{1}{2}\sum_{k=1}^t\lambda^{t-k}(y(t)-x(t)^\top \theta)^2+\dfrac{\lambda^t}{2}\big(\theta-\theta_0\big)^\top P_0^{-1} \big(\theta-\theta_0\big).  
\end{equation}
In Eq. \eqref{eq:Vt}, $\theta_0$ refers to a prior guess for the parameter vector, $P_0\succ 0$ is a symmetric  positive-definite weighting matrix reflecting the uncertainty related to the guess $\theta_0$, and $\lambda\in \interval[open]{0}{1}$ is a  forgetting factor which intends to downweight the information contained in the oldest data with respect to time $t$. 

\noindent Note that the objective function in \eqref{eq:Vt} is continuous, coercive and strictly convex, hence implying that the minimizer in \eqref{eq:optimization-RLS} exists and is unique. It can be shown that there exists a sequence of symmetric matrices\footnote{Indeed we have $P(t)=\left[\sum_{k=1}^t\lambda^{t-k}x(t)x(t)^\top+\lambda^t P_0^{-1} \right]^{-1}$ so that  $P^{-1}(t)=\lambda P^{-1}(t-1)+x(t)x(t)^\top$.} $\left\{P(t)\right\}$ such that the  solution $\theta(t)$ to the optimization problem \eqref{eq:optimization-RLS} can be recursively expressed as \cite{Goodwin09-Book}: 
\begin{align}
&\theta(t) = \theta(t-1)+q(t)(y(t)-x(t)^\top\theta(t-1)) \label{eq:update-theta}\\
&	q(t) = \dfrac{P(t-1)x(t)}{\lambda+x(t)^\top P(t-1)x(t)} \label{eq:update-q}\\
&P(t) = \dfrac{1}{\lambda}\big(P(t-1)-q(t)x(t)^\top P(t-1)\big) \label{eq:update-P}
\end{align}
where $\theta(0)=\theta_0$, $P(0)=P_0$.  Eqs \eqref{eq:update-theta}-\eqref{eq:update-P} define the well-known recursive least squares (point-valued) identifier with exponential forgetting factor \cite{Johnstone82-SCL}.  

\noindent For the purpose of the analysis to be presented in the sequel, define the parametric error $\tilde{\theta}(t)=\theta(t)-\theta^\circ$. 
Then it follows from the system equation \eqref{eq:system} and the $\theta$-update equation \eqref{eq:update-theta} that the error has the following dynamics 
\begin{equation}\label{eq:error}
	\tilde{\theta}(t)=A(t)\tilde{\theta}(t-1)+q(t)v(t), 
\end{equation}
with $A(t)=I_n-q(t)x(t)^\top$. Eq. \eqref{eq:error} together with \eqref{eq:update-q}-\eqref{eq:update-P} represents a dynamic system with input $\left\{v(t)\right\}$ and state $\{\tilde{\theta}(t)\}$. 
For future use in the paper, we can further express $\tilde{\theta}(t)$ in function of the initial error $\tilde{\theta}(0)$ and the noise sequence $\left\{v(k)\right\}_{1\leq k\leq t}$ up to time $t$,  
\begin{equation}\label{eq:convolution}
	\tilde{\theta}(t)=\Phi(t,0)\tilde{\theta}(0)+\sum_{j=1}^{t}\Phi(t,j)q(j)v(j), 
\end{equation}
where $\Phi$ is the state transition matrix defined by
\begin{equation}\label{eq:PHI}
	\Phi(t,t_0)=\left\{\begin{array}{lll}I_n & & t=t_0\\A(t)\cdots A(t_0+1) & & t>t_0\end{array}\right.
\end{equation}
An interesting property of the state transition matrix is that for any triplet  $(t,t_1,t_0)$ of nonnegative integers satisfying 
$t\geq t_1\geq t_0$, 
\begin{equation}\label{eq:Property-PHI}
	\Phi(t,t_0)=\Phi(t,t_1)\Phi(t_1,t_0). 
\end{equation}
Now we recall the stability concept which is of  interest in the following developments. 
For this purpose, consider the homogenous part of system \eqref{eq:error} (i.e., the one obtained when the input $v$ satisfies $v\equiv 0$), which we may generically describe by  
\begin{equation}\label{eq:LTV} 
	\xi(t)=A(t)\xi(t-1), \quad \xi(0)=\xi_0 
\end{equation}
where $A:\mathbb{Z}_+\rightarrow \Re^{n\times n}$ is a matrix-valued function and $\xi(t)\in \Re^n$ is the state of the system \eqref{eq:LTV}  at time $t\in \mathbb{Z}_+$. 
For any $(t,t_0)\in \mathbb{Z}_+$ with $t\geq t_0$, $\xi(t)$ can be related to $\xi(t_0)$ by 
$\xi(t)=\Phi(t,t_0)\xi(t_0). $
Using the generic LTV system \eqref{eq:LTV}, we now define the notion of exponential stability. 
\begin{definition}
The LTV system \eqref{eq:LTV} is said to be exponentially stable if there exist some constants $\gamma>0$ and $\rho\in \interval[open right]{0}{1}$ such that 
\begin{equation}\label{eq:stability-condition}
	\left\|\xi(t)\right\|_2\leq \gamma \rho^{t-t_0}\left\|\xi(t_0)\right\|_2 
\end{equation}
for all $(t,t_0)\in \mathbb{Z}_+$ such that $t\geq t_0$. Indeed \eqref{eq:stability-condition} is equivalent to $\left\|\Phi(t,t_0)\right\|_2\leq  \gamma \rho^{t-t_0}$.
\end{definition}
\noindent Finally, note that the non-homogenous system \eqref{eq:error} is  stable in the BIBO sense if \eqref{eq:LTV} is (exponentially) stable and the gain sequence $\left\{q(t)\right\}$ is bounded. We will see in the next section that such a property is guaranteed for the error system \eqref{eq:error} provided that the regressor $\left\{x(t)\right\}$  from the system \eqref{eq:system} is bounded and enjoys some richness condition.  

\subsection{A stability property for the RLS}
We first recall a definition of the concept of (uniform) persistence of excitation \cite{Johnstone82-SCL}. 
\begin{definition}\label{def:PE}
A vector-valued sequence $\left\{x(t)\right\}\subset \Re^n$ is said to be persistently exciting (PE) if there exist some   strictly positive constants $\alpha$ and $\beta$ (called excitation levels) and a time horizon $T$ such that 
\begin{equation}\label{eq:PE}
	\alpha I_n\preceq \sum_{k=t+1}^{t+T}x(k)x(k)^\top \preceq \beta I_n \quad \forall t\in \mathbb{Z}_+
\end{equation}
\end{definition}
\ifProof
\noindent The lower bound of \eqref{eq:PE} requires that the matrix of regressor $X_t\triangleq \bbm x(t+1) & \cdots & x(t+T)\eem $ be full rank on any time horizon of length $T$. Additionally, the smallest eigenvalue of $X_tX_t^\top$ must be larger than a minimum level $\alpha>0$. The upper bound in \eqref{eq:PE} expresses uniform boundedness of  the sequence   $\left\{x(t)\right\}$.
\fi
\begin{lem}\label{lem:bounded-InverseP}
Consider the RLS algorithm \eqref{eq:update-theta}-\eqref{eq:update-P} under the assumptions that $P(0)\succ 0$ and  $\lambda\in \interval[open]{0}{1}$. Then the matrices $P(t)$ defined by \eqref{eq:update-P} are invertible for all $t\in \mathbb{Z}_+$ and the  sequence  $\left\{P^{-1}(t)\right\}$ of their inverses satisfy 
\begin{equation}\label{eq:update-Inverse-P}
	P^{-1}(t)=\lambda P^{-1}(t-1)+x(t)x(t)^\top. 
\end{equation}
Moreover, if $\left\{x(t)\right\}$ is PE in the sense of Definition \ref{def:PE} with horizon $T$ and excitation levels $(\alpha,\beta)$, then $\left\{P^{-1}(t)\right\}$ is uniformly bounded as follows  
\begin{equation}\label{eq:boundedness-P(t)}
	\gamma_1I_n\preceq P^{-1}(t)\preceq \gamma_2 I_n \quad \forall t\geq 0
\end{equation}
with
\begin{align}
	&\gamma_1 =\min\left(\delta_1, \alpha \lambda^{2T-1}\right)\\
	& \gamma_2 = \max\big(\delta_2,\lambda^T\sigma_{\max}(P^{-1}(0))+\beta\dfrac{2-\lambda}{1-\lambda}\big)
\end{align}
and  $\delta_1=\min_{t=0,\ldots,T-1}\sigma_{\min}[P^{-1}(t)]$, $\delta_2=\max_{t=0,\ldots,T-1}\sigma_{\max}[P^{-1}(t)]$, $\sigma_{\min}[\cdot]$ and $\sigma_{\max}[\cdot]$ standing for the minimum and maximum eigenvalues respectively. 
\end{lem}
\noindent A proof of this lemma can be found in \cite{Bako16-Automatica-b}. 

\noindent Next we derive an input-to-state-stability (ISS) property for the error dynamics \eqref{eq:error} subject to  \eqref{eq:update-q}-\eqref{eq:update-P}. 
\begin{thm}\label{thm:ISS-RLS}
Consider the RLS algorithm applied to the data generated by system \eqref{eq:system}. If the regressor sequence $\left\{x(t)\right\}$ is PE, then 
\begin{equation}
	\big\|\tilde{\theta}(t)\big\|_2^2\leq \dfrac{1}{\gamma_1}\Big[\lambda^t \sigma_{\max}[P^{-1}(0)]\big\|\tilde{\theta}(0)\big\|_2^2+\sum_{k=1}^{t}\lambda^{t-k}v(k)^2 \Big]
\end{equation}
where $\tilde{\theta}(t)=\theta(t)-\theta^\circ$ is the parametric estimation error at time $t$ and $\gamma_1$ is any positive number satisfying \eqref{eq:boundedness-P(t)}. 
\end{thm}
\ifProof
{ 
\begin{proof}
Let  $V(t)=\tilde{\theta}(t)^\top P^{-1}(t)\tilde{\theta}(t)$. By subtracting the true parameter vector $\theta^\circ$ from each side of \eqref{eq:update-theta} and invoking the equation of the data-generating system \eqref{eq:system}, it is easy to see that $\tilde{\theta}(t)=\tilde{\theta}(t-1)+q(t)\varepsilon(t)$, 
where $\varepsilon(t)=y(t)-x(t)^\top \hat{\theta}(t-1)=v(t)-x(t)^\top \tilde{\theta}(t-1)$. On the other hand, we know from Lemma \ref{lem:bounded-InverseP} that $\left\{P^{-1}(t)\right\}$ obeys the recursive relation \eqref{eq:update-Inverse-P}. 
Now by direct algebraic calculations it can be seen that  
{
$$\begin{aligned}
	V(t)&=\big(\tilde{\theta}(t-1)+q(t)\varepsilon(t)\big)^\top \big(\lambda P^{-1}(t-1)+x(t)x(t)^\top\big) \times \ldots \\ & \hspace{4cm}\ldots \times \big(\tilde{\theta}(t-1)+q(t)\varepsilon(t)\big)\\
	& = \lambda V(t-1)+2\lambda\tilde{\theta}(t-1)^\top P^{-1}(t-1)q(t)\varepsilon(t)\\
	& \quad +2(x(t)^\top\tilde{\theta}(t-1))(x(t)^\top q(t))\varepsilon(t) + (x(t)^\top\tilde{\theta}(t-1))^2\\
	& \quad+\lambda q(t)^\top P^{-1}(t-1)q(t)\varepsilon(t)^2+(x(t)^\top q(t))^2\varepsilon(t)^2
\end{aligned}
$$
}
Note now that by posing $s(t)=\lambda+x(t)^\top P(t-1)x(t)$, we have 
$$\begin{aligned}
	&q(t)^\top P^{-1}(t-1)q(t)=\dfrac{1}{s(t)}-\dfrac{\lambda}{s(t)^2}\\
	& x(t)^\top q(t) = 1-\dfrac{\lambda}{s(t)}\\
	& P^{-1}(t-1)q(t) = \dfrac{x(t)}{s(t)}
\end{aligned}$$
Substituting these formulas in the above expression of $V(t)$ gives
$$ V(t)=\lambda V(t-1)-\dfrac{\lambda}{s(t)}\varepsilon(t)^2+v(t)^2.$$ 
It follows that $V(t)\leq \lambda V(t-1)+v(t)^2. $
Iterating this last equation and invoking the property \eqref{eq:boundedness-P(t)} of $\left\{P^{-1}(t)\right\}$  yields 
$$\gamma_1\big\|\tilde{\theta}(t)\big\|_2^2\leq  V(t)\leq \lambda^tV(0)+\sum_{k=1}^{t}\lambda^{t-k}v(k)^2  $$
which, by using the fact that $P^{-1}(0)\preceq \sigma_{\max}[P^{-1}(0)] I_n$,  implies that
$$\gamma_1\big\|\tilde{\theta}(t)\big\|_2^2\leq \lambda^t\sigma_{\max}[P^{-1}(0)]\big\|\tilde{\theta}(0)\big\|_2^2+\sum_{k=1}^{t}\lambda^{t-k}v(k)^2.  $$
Hence the claim of the theorem is established. 
\end{proof}
}
\else
{
\rev{
\begin{proof}
See \cite{Bako22} for a proof. 
\end{proof}
}
}
\fi

\begin{cor}\label{cor:Noise-free}
Under the conditions of Theorem \ref{thm:ISS-RLS}, if the noise $v$ of model \eqref{eq:system} is identically  equal to zero, then
\begin{equation}\label{eq:bound-norm-state}
	\big\|\tilde{\theta}(t)\big\|_2\leq \left(\dfrac{\lambda^{t}\sigma_{\max}[P^{-1}(0)]}{\gamma_1}\right)^{1/2}\big\|\tilde{\theta}(0)\big\|_2 
\end{equation}
that is, the sequence $\left\{\theta(t)\right\}$ generated by the RLS algorithm converges to $\theta^\circ$ exponentially fast regardless of the initial point $\theta(0)$.  
\end{cor}
\begin{lem}
Consider the state transition matrix-valued function $\Phi$ defined in \eqref{eq:PHI} from the RLS error system \eqref{eq:error}. If the regressor sequence $\left\{x(t)\right\}\subset\Re^n$ is PE, then there exist constant real positive numbers $\gamma_1$ and $\gamma_2$ such that for all $(t,t_0)$ obeying $t\geq t_0$, 
\begin{equation}\label{eq:Exponential-Decay-PHI}
	\left\|\Phi(t,t_0)\right\|_F\leq c \rho^{t-t_0} 
\end{equation}
where $c= (n\gamma_2 \gamma_1^{-1})^{1/2}$ and $\rho = \lambda^{1/2}$, $\lambda$ being the forgetting factor of the RLS algorithm. 
\end{lem}
\ifProof
{ 
\begin{proof}
Consider the error system \eqref{eq:error} under the assumption that the  noise sequence $\left\{v(t)\right\}$ is equal to zero. Then for any $(t,t_0)$ such that $t\geq t_0$, we have $\tilde{\theta}(t)=\Phi(t,t_0)\tilde{\theta}(t_0)$. 
Moreover, Corollary \ref{cor:Noise-free} can be applied by replacing the time origin for an arbitrary $t_0\in \mathbb{Z}_+$  such that $t\geq t_0\geq 0$. This gives
$$
\big\|\tilde{\theta}(t)\big\|_2\leq \left(\dfrac{\lambda^{t-t_0}\sigma_{\max}[P^{-1}(t_0)]}{\gamma_1}\right)^{1/2}\big\|\tilde{\theta}(t_0)\big\|_2 
$$
for any value of $\tilde{\theta}(t_0)\in \Re^n$. 
Since the PE condition holds here for $\left\{x(t)\right\}$, we know by Lemma \ref{lem:bounded-InverseP} that there exists  a constant number $\gamma_2>0$ such that $\sigma_{\max}[P^{-1}(t_0)]\leq \gamma_2$ (see Eq. \eqref{eq:boundedness-P(t)}). We can hence write
$$
\big\|\tilde{\theta}(t)\big\|_2\leq \left(\dfrac{\gamma_2}{\gamma_1}\right)^{1/2}\rho^{t-t_0}\big\|\tilde{\theta}(t_0)\big\|_2.  
$$
This implies that 
$$
\left\|\Phi(t,t_0)\right\|_2=\sup_{\tilde{\theta}(t_0)\neq 0}\dfrac{\big\|\Phi(t,t_0)\tilde{\theta}(t_0)\big\|_2}{\big\|\tilde{\theta}(t_0)\big\|_2} \leq \left(\dfrac{\gamma_2}{\gamma_1}\right)^{1/2}\rho^{t-t_0}. 
$$
Finally, the result follows by recalling that \\ $\left\|\Phi(t,t_0)\right\|_F\leq \sqrt{n}\left\|\Phi(t,t_0)\right\|_2$. 
\end{proof}
}
\else
{
\rev{
\begin{proof}
See \cite{Bako22} for a proof. 
\end{proof}
}
}
\fi

\section{Interval-valued estimator}\label{sec:interval-estimator}
\rev{\noindent In this section we present the main contributions of the paper concerning the development of an adaptive interval-valued parametric estimator.  As explained at the end of Section \ref{sec:Problem}, our method relies  on the error sequence generated by a point-value adaptive estimator. Considering the special case of the RLS, we obtain the error dynamics expressed in \eqref{eq:convolution} which is directly related to the noise. Applying Lemma \ref{lem:Az+w} to this equation gives an interval estimate of the error $\tilde{\theta}(t)=\hat{\theta}(t)-\theta^\circ$, which can then be modified to get an estimate of $\theta^\circ$. 
}

\subsection{Derivation of an interval-valued estimator}
Assume now that we are given a recursive point-valued estimator (say the RLS algorithm dicussed earlier) generating a sequence of estimates $\left\{\theta(t)\right\}$ for $\theta^\circ$ in \eqref{eq:system}.  
To derive an interval-valued estimator for $\theta^\circ$, 
we first find an interval-valued estimate for the error $\tilde{\theta}(t)$ defined in \eqref{eq:error}. We do so by applying Lemma \ref{lem:Az+w} to \eqref{eq:convolution} which, for convenience, can be rewritten as
$ \tilde{\theta}(t)=M(t)z(t)$
with 
$$
\begin{aligned}
	& M(t)=\begin{bmatrix}\Phi(t,0) & \Phi(t,1)q(1) & \cdots & \Phi(t,t)q(t)\end{bmatrix}\\
	& z(t)=\begin{bmatrix}\tilde{\theta}(0)^\top & v(1) & \cdots & v(t)\end{bmatrix}^\top. 
\end{aligned}
$$
We hence obtain immediately from  Lemma \ref{lem:Az+w} that the  smallest interval set containing the parametric error $\tilde{\theta}(t)$  can be expressed in term of its center-radius pair $(c_{\tilde{\theta}},r_{\tilde{\theta}})$ given by 
\begin{align}
&c_{\tilde{\theta}}(t)=\Phi(t,0)c_{\tilde{\theta}}(0)+\sum_{j=1}^{t}\Phi(t,j)q(j)c_v(j) \label{eq:c-theta}\\
& r_{\tilde{\theta}}(t)=|\Phi(t,0)|r_{\tilde{\theta}}(0)+\sum_{j=1}^{t}|\Phi(t,j)q(j)|r_v(j) \label{eq:r-theta}
\end{align}
where $(c_v,r_v)$ is the pair of signals defining the intervals of the noise sequence $\left\{v(t)\right\}$ and $c_{\tilde{\theta}}(0)=c_{\theta}(0)-\theta^\circ$ and $r_{\tilde{\theta}}(0)=r_{\theta}(0)$. 
Recalling now that $\theta^\circ = \theta(t)-\tilde{\theta}(t)$, an interval-valued estimate of the $\theta^\circ$ can be obtained as proposed in the following proposition.

\begin{prop}\label{prop:INT}
Consider the system \eqref{eq:system} and assume that the regressor sequence $\left\{x(t)\right\}$ is  PE in the sense of Definition \ref{def:PE} and that the noise $\left\{v(t)\right\}$ is bounded and admits an interval representation $(c_v(t),r_v(t))$.  
Then the  pair $(c_\theta(t),r_\theta(t))$ given by 
\begin{equation}\label{eq:interval-estimate}
	\left\{\begin{aligned}
		&c_\theta(t)=\theta(t)-c_{\tilde{\theta}}(t)\\
		&r_{\theta}(t)=r_{\tilde{\theta}}(t)
	\end{aligned}\right.
\end{equation}
with $(c_{\tilde{\theta}}(t),r_{\tilde{\theta}}(t))$ as in \eqref{eq:c-theta}-\eqref{eq:r-theta}, defines an interval estimator for the parameter vector $\theta^\circ$.
\end{prop}
\begin{proof}
To begin with, let us recall that $\theta^\circ = \theta(t)-\tilde{\theta}(t)$ and \eqref{eq:c-theta}-\eqref{eq:r-theta} is such that  $\tilde{\theta}(t)\in \interval{c_{\tilde{\theta}}(t)-r_{\tilde{\theta}}(t)}{c_{\tilde{\theta}}(t)+r_{\tilde{\theta}}(t)}$ for all $t\in \mathbb{Z}_+$.   It follows, as an immediate consequence, that
$\theta^\circ\in  \interval{c_\theta(t)-r_\theta(t)}{c_\theta(t)+r_\theta(t)}$ with $(c_\theta,r_\theta)$ defined as in \eqref{eq:interval-estimate}.  To reach the conclusion of the proposition, it remains  to prove that the dynamical systems (operators) $(c_\theta(0),c_v)\mapsto c_\theta$ and $(r_\theta(0),r_v)\mapsto r_\theta$ are BIBO stable. By relying on \eqref{eq:c-theta}-\eqref{eq:r-theta}, it is immediate to see that, under the PE condition, both  properties follow indeed from \eqref{eq:Exponential-Decay-PHI} which in turn  is a consequence of Theorem \ref{thm:ISS-RLS}. 
\end{proof}

\subsection{Computational aspects}
Implementing numerically the estimator \eqref{eq:interval-estimate} requires computing $(c_{\tilde{\theta}}(t),r_{\tilde{\theta}}(t))$ defined in \eqref{eq:c-theta}-\eqref{eq:r-theta} for any time $t$. A problem however is that these convolutional formulas become infeasible in practice when $t$ grows towards infinity. Therefore it is desirable to find an efficient implementation of this estimator for example, in the form of a one-step-ahead state-space recursive realization. 
In this perspective, note that $c_\theta(t)$ can be computed recursively through the following equation
\begin{equation}\label{eq:realization-c}
c_{\theta}(t)=A(t) c_{\theta}(t-1)+q(t)\big(y(t)-c_v(t)\big). 
\end{equation}
Unfortunately, there is, in general, no simple recursive implementation for the interval radius $r_{\theta}$ as defined \eqref{eq:r-theta}-\eqref{eq:interval-estimate}. Hence, a strategy would be to search for a more pessimistic estimate $r_{\theta}$ but which would be implementable. That is, the computational constraint introduces a dose of pessimism in the estimation, resulting in a less tight interval-valued estimate of $\theta^\circ$. 
A possible solution is to replace $r_{\theta}$ with a truncated version $\hat{r}_{\theta,m}$ defined, for a given integer $m>0$, by 
\begin{equation}\label{eq:rhat}
		\hat{r}_{\theta,m}(t)=
	\left\{	\begin{array}{ll} 
	\displaystyle	\left|\Phi(t,0)\right|{r}_{\theta}(0)+ \sum_{k=1}^{t} |\Phi(t,k)q(k)|r_v(k) \\  \mbox{ if } t=0,\ldots,m\\
	\displaystyle	\left|\Phi(t,t-m)\right|\hat{r}_{\theta,m}(t-m)+ \sum_{k=t-m+1}^{t} |\Phi(t,k)q(k)|r_v(k), \\ \mbox{ if } t> m
		\end{array}\right.
\end{equation}
Intuitively $\hat{r}_{\theta,m}$ is all the smaller as $m$ is large. On the other  hand the computational complexity grows with $m$. Note that in the extreme case where  $m=t$, we get $\hat{r}_{\theta,m}(t)=r_\theta(t)$ for all $t\in \mathbb{Z}_+$.   The simplest version (but also the most pessimistic) of the family \eqref{eq:rhat} of estimates is obtained for $m=1$,  
\begin{equation}
		\hat{r}_{\theta,1}(t)= \left|A(t)\right|\hat{r}_{\theta,1}(t-1)+ |q(t)|r_v(t)
\end{equation}
\rev{However, as we will see shortly, such an estimate is unlikely to satisfy the BIBO condition of Definition \ref{def:Interval-Estimator}. As a consequence it will not qualify in general as an interval-valued estimator
}

The result below formally shows that for any time $t$, the interval $\mathscr{I}(c_\theta(t,r_{\theta}(t))$ is included in  $\mathscr{I}(c_\theta(t),\hat{r}_{\theta,m}(t))$.  
\begin{lem}\label{lem:r<=r'}
Consider the interval radii $r_{\theta}(t)$ and $\hat{r}_{\theta,m}(t)$  defined in \eqref{eq:r-theta} and \eqref{eq:rhat} respectively. 
For any fixed integer $m$, it holds that $r_{\theta}(t)\leq \hat{r}_{\theta,m}(t)$ for all $t\in \mathbb{Z}_+$. 
\end{lem}
\begin{proof}
We start by observing that $\hat{r}_{\theta,m}(t)=r_{\theta}(t)$ for $t=0,\ldots,m$. Hence the inequality is true for $t=0,\ldots,m$. If $t>m$, write $t=\alpha(t)m+\beta(t)$ for some positive integers $(\alpha(t),\beta(t))$ such that $0\leq \beta(t)<m$. By iterating the second equation of  \eqref{eq:rhat}, we ultimately get 
$$\begin{aligned}
	\hat{r}_{\theta,m}(t)=&\Big(\prod_{\ell=1}^{\alpha(t)}\Big|\Phi\big(t-(\ell-1)m,t-\ell m\big)\Big|\Big) \hat{r}_{\theta,m}(\beta(t))\\
	&+\sum_{j=1}^{\alpha(t)}\: \sum_{k=t-jm+1}^{t-(j-1)m}\: \Big(\prod_{\ell=1}^{j-1}\Big|\Phi\big(t-(\ell-1)m,t-\ell m\big)\Big|\Big)\times \ldots\\
	 & \hspace{2.2cm}  \ldots \times \Big|\Phi\big(t-(j-1)m,k\big)q(k)\Big|r_v(k)
\end{aligned}
$$
Moreover, since $\beta(t)=t-\alpha(t)m$ satisfies $0\leq \beta(t)<m$, by the definition \eqref{eq:rhat} of $\hat{r}_{\theta,m}(t)$ we know that 
$$\hat{r}_{\theta,m}(\beta(t))=\left|\Phi(t-\alpha(t)m,0)\right|{r}_{\theta}(0)+ \sum_{k=1}^{t-\alpha(t)m} |\Phi(t-\alpha(t)m,k)q(k)|r_v(k).$$
 Plugging this in the above expression yields
\begin{equation}\label{eq:expression-Rtheta}
\begin{aligned}
	\hat{r}_{\theta,m}(t)=&\Big(\prod_{\ell=1}^{\alpha(t)+1}\Big|\Phi\big(\kappa(t,\ell-1),\kappa(t,\ell)\big)\Big|\Big) {r}_{\theta}(0)\\
	&+\sum_{j=1}^{\alpha(t)+1}\: \sum_{\kappa(t,j)+1}^{\kappa(t,j-1)}\: \Big(\prod_{\ell=1}^{j-1}\Big|\Phi\big(\kappa(t,\ell-1),\kappa(t,\ell)\big)\Big|\Big)\times \ldots \\
	& \hspace{2.2cm}  \ldots\times \Big|\Phi\big(\kappa(t,j-1),k\big)q(k)\Big|r_v(k)
\end{aligned}
\end{equation}
where $\kappa(t,j)=\max(t-jm,0)$. 
Using the property $|A||B|\geq |AB|$, we observe, for example, that the matrix in the first term of \eqref{eq:expression-Rtheta} can be bounded as follows
$$
\begin{aligned}
	 &\prod_{\ell=1}^{\alpha(t)+1}\Big|\Phi\big(\kappa(t,\ell-1),\kappa(t,\ell)\big)\Big|\\
	 &=\Big(\prod_{i=1}^{\alpha(t)}\Big|\Phi\big(t-(i-1)m,t-im\big)\Big|\Big) \left|\Phi(t-\alpha(t)m,0)\right|\\
	& \geq \Big|\Phi\big(t,t-m\big) \Phi\big(t-m,t-2m\big)\cdots \Phi\big(t-\alpha(t)m,0\big)\Big|\\
	& = \big|\Phi\big(t,0\big)\big|,
\end{aligned}
$$
where the last equality follows from the property \eqref{eq:Property-PHI} of the state transition matrix. Applying this property to the whole expression of  $\hat{r}_{\theta,m}(t)$ above leads to 
$$\begin{aligned}
	\hat{r}_{\theta,m}(t)&\geq  \big|\Phi\big(t,0\big)\big| r_{\theta}(0)
	+\sum_{j=1}^{\alpha(t)}\: \sum_{k=t-jm+1}^{t-(j-1)m}\: \big|\Phi\big(t,k\big)q(k)\big|r_v(k)\\
	& \hspace{2.5cm} +\sum_{k=1}^{t-\alpha(t)m} |\Phi(t,k)q(k)|r_v(k)\\
	& = \big|\Phi\big(t,0\big)\big| r_{\theta}(0)	+\sum_{k=1}^{t}\: \big|\Phi\big(t,k\big)q(k)\big|r_v(k)\\
	& = r_{\theta}(t)
\end{aligned}
$$
This concludes the proof. 
\end{proof}

\begin{thm}\label{thm:m-star}
Consider the system \eqref{eq:system} and assume that the regressor sequence $\left\{x(t)\right\}$ is  PE in the sense of Definition \ref{def:PE} and that the noise $\left\{v(t)\right\}$ is bounded and admits an interval representation $(c_v(t),r_v(t))$.  
Then there exists a number $m^\star>0$ such that for all $m\geq m^\star$, $\big(c_{\theta}(t), \hat{r}_{\theta,m}(t)\big)$ defined in \eqref{eq:interval-estimate} and \eqref{eq:rhat} constitutes an interval estimator for the parameter vector $\theta^\circ$.
\end{thm}
\begin{proof}
By Lemma \ref{lem:r<=r'}, we know that $r_{\theta}(t)\leq \hat{r}_{\theta,m}(t)$. Hence it is immediate by Proposition \ref{prop:INT} that $\theta^\circ\in \interval{c_{\theta}(t)-\hat{r}_{\theta,m}(t)}{c_{\theta}(t)+\hat{r}_{\theta,m}(t)} \: \forall t\in \mathbb{Z}_+.$ We just need to establish the second condition of Definition \ref{def:Interval-Estimator}. More precisely, we need to show that the sequence of intervals $\interval{c_{\theta}(t)-\hat{r}_{\theta,m}(t)}{c_{\theta}(t)+\hat{r}_{\theta,m}(t)}$ is bounded. For this purpose we will show that the systems  $(c_{\theta}(0),c_v)\mapsto c_{\theta}$ and $(r_{\theta}(0),r_v)\mapsto \hat{r}_{\theta,m} $ are stable. 

\noindent\textit{Boundedness of $\left\{c_\theta(t)\right\}$}. 
Recall that a state-space realization of $c_\theta$ is given in \eqref{eq:realization-c}. Relying on this equation, it is clear that $(c_{\theta}(0),c_v)\mapsto c_{\theta}$ is BIBO state if and only if $(\tilde{\theta}(0),v)\mapsto \tilde{\theta}$ defined in \eqref{eq:error} is BIBO stable.  As already shown in the proof of Theorem \ref{thm:ISS-RLS}, the PE condition suffices to guarantee that $\tilde{\theta}$ is bounded whenever $(\tilde{\theta}(0),v)$ is bounded. We hence conclude that $(c_{\theta}(0),c_v)\mapsto c_{\theta}$ is BIBO. 

\noindent \textit{Boundedness of $\left\{\hat{r}_{\theta,m}(t)\right\}$}. 
We will rely on formula \eqref{eq:expression-Rtheta}. Taking the Frobenius norm of $\hat{r}_{\theta,m}(t)$ and applying 
 the submultiplicativity property of the Frobenius norm and the fact that $\left\||A|\right\|_F = \left\|A\right\|_F$ (i.e., Frobenius norm of $A$ equals Frobenius norm of $|A|$) for any matrix $A$, we can write
$$ \begin{aligned}
	\left\|\hat{r}_{\theta,m}(t)\right\|_2
	&\leq \Big(\prod_{i=1}^{\alpha(t)+1}\Big\|\Phi\big(\kappa(t,i-1),\kappa(t,i)\big)\Big\|_F\Big) \left\|{r}_{\theta}(0)\right\|_2\\
	&+\sum_{j=1}^{\alpha(t)+1}\!\! \sum_{k=\kappa(t,j)+1}^{\kappa(t,j-1)}\: \Big(\prod_{\ell=1}^{j-1}\Big\|\Phi\big(\kappa(t,\ell-1),\kappa(t,\ell)\big)\Big\|_F\Big)\times \ldots\\
	& \hspace{1.8cm}\ldots \times \Big\|\Phi\big(\kappa(t,j-1),k\big)\Big\|_F\left\|q(k)\right\|_2 \left\|r_v(k)\right\|_2
\end{aligned}$$ 
Here, we have used the fact that  $\left\|x\right\|_F=\left\|x\right\|_2$ for any vector $x$. 
Since the PE condition holds, it follows from the analysis of Section \ref{sec:ref-adaptive-identifier} (See Eq. \eqref{eq:Exponential-Decay-PHI}) that the transition matrix $\Phi$ satisfies 
$\left\|\Phi(t,t_0)\right\|_F\leq c\rho^{t-t_0}$ with the  constants $c>0$ and $\rho$ being  defined as in \eqref{eq:Exponential-Decay-PHI}. 
Applying this in the above inequality  gives
$$
\begin{aligned}
	\left\|\hat{r}_{\theta,m}(t)\right\|_2&\leq  (c\rho^m)^{\alpha(t)}(c\rho^{t-\alpha(t)m})\left\|{r}_{\theta}(0)\right\|_2\\
	& +\sum_{j=1}^{\alpha(t)}\: \sum_{k=t-jm+1}^{t-(j-1)m}(c\rho^m)^{j-1} (c\rho^{t-(j-1)m-k})\left\|q(k)\right\|_2\left\|r_v(k)\right\|_2\\
	& + \sum_{k=1}^{t-\alpha(t)m} (c\rho^m)^{\alpha(t)}(c\rho^{t-\alpha(t)m-k})\left\|q(k)\right\|_2  \left\|r_v(k)\right\|_2
\end{aligned}
$$  
\noindent Under the PE condition of $\left\{x(t)\right\}$, we know by Lemma \ref{lem:bounded-InverseP} that $\left\{P(t)\right\}$ is uniformly bounded as $1/\gamma_2I_n\preceq P(t)\preceq 1/\gamma_1I_n$ for all $t$. From this, it is easy to see that the vector $q(t)$ defined  in \eqref{eq:update-q} satisfies 
\begin{equation}\label{eq:bound-q}
	\left\|q(k)\right\|_2\leq \dfrac{1/\gamma_1 \sup_{t}\left\|x(t)\right\|_2}{\lambda+1/\gamma_2\inf_t\left\|x(t)\right\|_2^2}. 
\end{equation}
 This implies  that $\left\{q(t)\right\}$ is upper-bounded. On the other hand, $r_v$ is bounded by assumption.     
Let therefore consider the bounds $\eta_q=\sup_{k\in \mathbb{Z}_+} \left\|q(k)\right\|_2$ and $\eta_v = \sup_{k\in \mathbb{Z}_+} \left\|r_v(k)\right\|_2$.  Using these notations and proceeding from above gives 
$$
\begin{aligned}
		\left\|\hat{r}_{\theta,m}(t)\right\|_2\leq  (c\rho^m)^{\alpha(t)}c\left\|{r}_{\theta}(0)\right\|_2
	&+c \eta_q\eta_v \dfrac{1-\rho^{m}}{1-\rho}\dfrac{1-(c\rho^m)^{\alpha(t)}}{1-c\rho^m}\\
	 &+ c \eta_q\eta_v (c\rho^m)^{\alpha(t)}\dfrac{1-\rho^{t-\alpha(t)m}}{1-\rho}
\end{aligned}
$$  
This inequality can be refined as 
\begin{equation}
\begin{aligned}
		\left\|\hat{r}_{\theta,m}(t)\right\|_2	\leq  &(c\rho^m)^{\alpha(t)}c\left\|{r}_{\theta}(0)\right\|_2\\
		&+\dfrac{c \eta_q\eta_v}{1-\rho} \left[\dfrac{1-\rho^{m}}{1-c\rho^m}+  (c\rho^m)^{\alpha(t)}(1-\rho^{m-1})\right].
\end{aligned}
\end{equation}
Note that\footnote{$\lfloor \cdot\rfloor$ refers to the floor function (integer part).} $\alpha(t)=\left\lfloor t/m\right\rfloor \rightarrow +\infty$ as $t\rightarrow +\infty$. 
Hence, if $c\rho^m<1$, that is, if $m>m^\star \triangleq -\dfrac{\ln(c)}{\ln(r)}$, then the sequence $\left\{\hat{r}_{\theta,m}(t)\right\}$ is bounded. 
\end{proof}

A candidate for the constant $c$  is the one expressed in \eqref{eq:Exponential-Decay-PHI}. 
By making use of it, a full expression of $m^\star$ can be obtained as
	\begin{equation}\label{eq:m-star}
		m^\star =  -\dfrac{\ln\left(n\gamma_2\gamma_1^{-1}\right)}{\ln(\lambda)}
	\end{equation}
This suggests  that the richer the regressor sequence $\left\{x(t)\right\}$ (that is, the smaller the ratio $\gamma_1/\gamma_2$), the smaller the threshold $m^\star$ will be. Note indeed that $\gamma_1$ and $\gamma_2$ depend not only on the data sequence $\left\{x(t)\right\}$ but also on the forgetting factor $\lambda$ and the initial weighting matrix $P^{-1}(0)$.  
A few further comments can be made concerning the behavior of $\hat{r}_{\theta,m}$. 
First, note that an asymptotic bound on the estimated interval radius can be derived as follows
\begin{equation}
	\limsup_{t\rightarrow+\infty}\left\|\hat{r}_{\theta,m}(t)\right\|_2\leq \dfrac{c \eta_q\eta_v}{1-\rho} \dfrac{1-\rho^{m}}{1-c\rho^m}. 	
\end{equation}
Then we see that as the truncation order $m$ grows, the asymptotic bound on $\hat{r}_{\theta,m}(t)$ gets closer to $b_{\infty}^\star\triangleq c\eta_q\eta_v/(1-\rho)$. 
By invoking Eq. \eqref{eq:bound-q} it is immediate to see that if we let $h_{\min}=\inf_{t}\left\|x(t)\right\|_2$ and $h_{\max}=\sup_{t}\left\|x(t)\right\|_2$,  then 
$\eta_q\leq \frac{\gamma_2}{\gamma_1} \frac{h_{\max}}{h_{\min}^2+\lambda \gamma_2}$
which, by using the expressions of $c$ and $\rho$ given in \eqref{eq:Exponential-Decay-PHI},  implies that 
\begin{equation}
	b_{\infty}^\star\leq \dfrac{\eta_v n^{1/2}}{1-\lambda^{1/2}} \left(\dfrac{\gamma_2}{\gamma_1}\right)^{3/2}\dfrac{h_{\max}}{h_{\min}^2+\lambda\gamma_2}.  
\end{equation}
What this shows is that the influencing parameters of  the bound $b_{\infty}^\star$ originates from three sources: (i) 
 the parameters  measuring  richness  of the learning data: $\gamma_1, \gamma_2, h_{\min}, h_{\max}$; (ii)
	the design parameters of the estimator: $\lambda$, $P^{-1}(0)$;
	(iii) the magnitude  $\eta_v$ of the uncertainty associated with the mathematical representation \eqref{eq:system} of the data.  
%

\subsection{Further improvements }\label{subsec:monotonic}
Due to the presence of noise in the data, the size of the interval estimates $(c_\theta,r_{\theta})$ or $(c_\theta,\hat{r}_{\theta,m})$ discussed above may oscillate over time instead of decreasing monotonically (See Figure \ref{fig:Int} for a visual illustration of this phenomenon). This behavior is undesirable  in practice and should be mitigated  as much as possible.  
For this purpose, we discuss here a simple recursive intersection operation for removing such possible non monotonic trend of the interval-valued estimate for the estimators proposed in the previous sections. To this end, consider a pair $(\underline{\xi},\overline{\xi}):\mathbb{Z}_+\rightarrow\Re^n\times \Re^n$  such that the to-be-estimated parameter vector $\theta^\circ$ lies in $\interval{\underline{\xi}(t)}{\overline{\xi}(t)}$ for all $t$. 
%
Define the pair of vector-valued functions $(\underline{\theta},\overline{\theta}):\mathbb{Z}_+\rightarrow\Re^n\times \Re^n$   such that $\theta^\circ\in \interval{\underline{\theta}(0)}{\overline{\theta}(0)}$ and for all $t\geq 1$, 
\begin{align}
& \underline{\theta}(t)=\max\big(\underline{\theta}(t-1),\underline{\xi}(t)\big)\label{eq:theta-inf} \\
& \overline{\theta}(t)=\min\big(\overline{\theta}(t-1),\overline{\xi}(t)\big), \label{eq:theta-sup}
\end{align}
where the minimum/maximum operators apply componentwise, i.e., when $x$ and $y$ are vectors of the same dimension,  $\min(x,y)$ refers to the vector whose entries are given by $\min(x_i,y_i)$. We will call $(\underline{\xi},\overline{\xi})$ the input of the dynamic system \eqref{eq:theta-inf}-\eqref{eq:theta-sup} and $(\underline{\theta},\overline{\theta})$ its state. In fact \eqref{eq:theta-inf}-\eqref{eq:theta-sup} is equivalent to $\interval{\underline{\theta}(t)}{\overline{\theta}(t)}=\interval{\underline{\theta}(t-1)}{\overline{\theta}(t-1)}\cap \interval{\underline{\xi}(t)}{\overline{\xi}(t)}$. 

\noindent We now state some basic properties of the estimator \eqref{eq:theta-inf}-\eqref{eq:theta-sup}. 

\begin{lem}\label{lem:Monotonic}
Assume $\underline{\theta}(0)\leq \overline{\theta}(0)$  and $\underline{\xi}(t)\leq \overline{\xi}(t)$ for all $t$. Then the following facts are true: 
\begin{enumerate}
	\item Boundedness: $\underline{\theta}(0)\leq \underline{\theta}(t)\leq \overline{\theta}(t)\leq \overline{\theta}(0)$ $\forall t\geq 0$
	\item  Monotonically decreasing widths: $\interval{\underline{\theta}(t)}{\overline{\theta}(t)}\subset \interval{\underline{\theta}(k)}{\overline{\theta}(k)}$ $\forall (k,t)$ such that $k\leq t$. 
	\item Convergence: The sequences $\left\{\underline{\theta}(t)\right\}$ and $\left\{\overline{\theta}(t)\right\}$ converge to $\underline{\theta}^*$ and $\overline{\theta}^*$ respectively with 
	$$
	\begin{aligned}
		&\underline{\theta}^*\triangleq \max\big(\underline{\theta}(0),\max_t\underline{\xi}(t)\big)\\
		&\overline{\theta}^*\triangleq \min\big(\overline{\theta}(0),\min_t\overline{\xi}(t)\big)
	\end{aligned}
	$$
	If $\max_t\underline{\xi}(t)\leq \underline{\theta}(0)$ and $\min_t\overline{\xi}(t)\leq \overline{\theta}(0)$, then the input sequence $\left\{\underline{\xi}(t),\overline{\xi}(t)\right\}$ does not bring any information since in this case $\underline{\theta}(t)=\underline{\theta}(0)$ and $\overline{\theta}(t)=\overline{\theta}(0)$ for all $t$.
	\item If $\theta^\circ\in \interval{\underline{\theta}(0)}{\overline{\theta}(0)}\cap \interval{\underline{\xi}(t)}{\overline{\xi}(t)}$ for all $t$, then \\$\theta^\circ \in \bigcap_{t=0}^{\infty}\interval{\underline{\theta}(t)}{\overline{\theta}(t)}$. 
\end{enumerate}
\end{lem}
\begin{proof}
The facts 1, 2 and 4 are quite immediate. To see why fact 3 holds, note that $\left\{\underline{\theta}(t)\right\}$ is (componentwise) nondecreasing and upper-bounded while $\left\{\overline{\theta}(t)\right\}$ is nonincreasing and lower-bounded. Hence by the monotone convergence theorem, both sequences are convergent and their limits are the maximal element $\underline{\theta}^*$ and minimal element $\overline{\theta}^*$ of the respective sequences as expressed above.  
\end{proof}
\begin{rem}
In virtue of the properties stated in Lemma \ref{lem:Monotonic}, the estimator in \eqref{eq:theta-inf}-\eqref{eq:theta-sup} is naturally robust to potential outliers in the sequence $\left\{(\underline{v}(t),\overline{v}(t))\right\}$ of bounds on the equation errors in \eqref{eq:system}.  
\end{rem}

\section{Application to a time-varying system}\label{sec:LTV}
We now consider the case where the true parameter vector $\theta^\circ$ in \eqref{eq:system} is no longer constant but may be time varying with a limited rate of change.  
Let us pose 
 \begin{equation}\label{eq:dynamics-theta0}
	 \theta^\circ(t)=\theta^\circ(t-1)+\delta(t),
 \end{equation}
where $\left\{\delta(t)\right\}$ is unknown but assumed to be bounded in an interval. More precisely, we assume that we know a sequence $\left\{\mathscr{I}(c_\delta(t),r_\delta(t))\right\}$ of intervals such that 
$\delta(t)\in \mathscr{I}(c_\delta(t),r_{\delta}(t))$ for all $t\in \mathbb{Z}_+$.  
Let us still use the notation $\tilde{\theta}(t)$ to refer to the parametric error now defined by 
 $\tilde{\theta}(t)=\theta(t)-\theta^\circ(t)$ with $\theta(t)$ generated as in \eqref{eq:update-theta} from the data. 
It can then be shown that the error dynamics take the form  
\begin{equation}\label{eq:theta-tilde}
\tilde{\theta}(t) = A(t)\tilde{\theta}(t-1) +B(t)\bar{v}(t)
\end{equation}
with $A(t)=I_n-q(t)x(t)^\top$ as in \eqref{eq:error} and 
\begin{equation}\label{eq:v-bar}
\begin{aligned}
		&B(t)=\bbm q(t) & -A(t)\eem \\
		&\bar{v}(t)=\bbm v(t) & \delta(t)^\top \eem^\top 
\end{aligned}
\end{equation}
Note in passing that one recovers the error dynamics \eqref{eq:error} from \eqref{eq:theta-tilde}  when $\delta(t)=0$ for all $t$, that is, when $\theta^\circ(t)$ is assumed constant. 
Now an interval representation of  $\bar{v}(t)$ in \eqref{eq:theta-tilde} is given by
\begin{equation}
	\left\{\begin{aligned}
		& c_{\bar{v}}(t)  = \bbm c_v(t) & c_{\delta}(t)\eem^\top \\
		& r_{\bar{v}}(t)=\bbm r_v(t) & r_{\delta}(t)\eem^\top.  
	\end{aligned}\right. 
\end{equation}
The relation \eqref{eq:theta-tilde} is key for deriving an interval-valued estimator. 
In effect, by relying on it and following the preceding discussions, it is easy to obtain, under the PE condition, an interval-valued estimator for the vector-valued sequence $\left\{\theta^\circ(t)\right\}$. 
More precisely, the complete form of the estimator is  $\mathscr{I}(c'_\theta(t),r'_{\theta}(t))=\interval{c'_{\theta}(t)-r'_{\theta}(t)}{c'_{\theta}(t)+r'_{\theta}(t)}$, with center $c'_\theta$ defined by the state-space equation
\begin{equation}\label{eq:c-LTV}
	c'_{\theta}(t)=A(t)c'_{\theta}(t-1)+q(t)\big(y(t)-c_v(t)\big)+A(t)c_{\delta}(t), 
\end{equation}
$	c'_{\theta}(0)=c_\theta(0)$, and radius $r'_\theta$ given in convolution form by 
\begin{equation}\label{eq:r-LTV}
		 r'_{\theta}(t)=|\Phi(t,0)|r_{\theta}(0)+\sum_{j=1}^{t}|\Phi(t,j)B(j)|r_{\bar{v}}(j).  
\end{equation}
Recall that in \eqref{eq:c-LTV},  $\left\{\theta(t)\right\}$ still refers to the sequence generated by the point-valued RLS identifier \eqref{eq:update-theta}-\eqref{eq:update-P}. Likewise, $\Phi$ is the RLS transition matrix expressed in \eqref{eq:PHI}. 
As to the truncated form of the estimator, it now admits the expression $\mathscr{I}(c'_\theta(t),\hat{r}'_{\theta,m}(t))=\interval{c'_{\theta}(t)-\hat{r}'_{\theta,m}(t)}{c'_{\theta}(t)+\hat{r}'_{\theta,m}(t)}$
with $c'_\theta$ as in \eqref{eq:c-LTV} and $\hat{r}_{\theta,m}(t)$ defined by 
\begin{equation}\label{eq:rhat-2}
		\hat{r}'_{\theta,m}(t)=
	\left\{	\begin{array}{ll} 
	\displaystyle	\left|\Phi(t,0)\right|{r}_{\theta}(0)+ \sum_{k=1}^{t} |\Phi(t,k)B(k)|r_{\bar{v}}(k), \\
	 \mbox{ if } t=0,\ldots,m\\
	\displaystyle	\left|\Phi(t,t-m)\right|\hat{r}'_{\theta,m}(t-m)+ \sum_{k=t-m+1}^{t} |\Phi(t,k)B(k)|r_{\bar{v}}(k), \\ \mbox{ if } t> m
		\end{array}\right.
\end{equation}

\noindent Finally, let us remark that it is possible, similarly as in Section \ref{subsec:monotonic}, to derive improved  versions of the above interval-valued estimators for the case of time-varying systems. For this purpose, consider any pair $(\underline{\xi},\overline{\xi})$  of functions such that $\underline{\xi}(t)\leq \overline{\xi}(t)$ and $\theta^\circ(t)\in \interval{\underline{\xi}(t)}{\overline{\xi}(t)}$ for all $t\geq 0$. Then by letting $(\underline{p},\overline{p})$ be defined by 
\begin{align}
& \underline{p}(t)=\max\big(\underline{p}(t-1)+\underline{\delta}(t),\underline{\xi}(t)\big)\label{eq:theta-inf2} \\
& \overline{p}(t)=\min\big(\overline{p}(t-1)+\overline{\delta}(t),\overline{\xi}(t)\big), \label{eq:theta-sup2}
\end{align}
with $\underline{\delta}(t)=c_{\delta}(t)-r_{\delta}(t)$ and $\overline{\delta}(t)=c_{\delta}(t)+r_{\delta}(t)$, it holds that  $\theta^\circ(t)\in \interval{\underline{p}(t)}{\overline{p}(t)}$ for all $t\geq 0$  provided that $\theta^\circ(0)\in \interval{\underline{p}(0)}{\overline{p}(0)}$. Moreover, $(\underline{p},\overline{p})$ is bounded provided that $(\underline{\xi},\overline{\xi})$ is bounded. Of course the inputs  $(\underline{\xi},\overline{\xi})$  of  \eqref{eq:theta-inf2}-\eqref{eq:theta-sup2} can be taken to be any of the  estimates $\big(c'_\theta-r'_\theta,c'_\theta-r'_\theta\big)$  in \eqref{eq:c-LTV}-\eqref{eq:r-LTV} or $\big(c'_\theta-\hat{r}'_{\theta,m},c'_\theta-\hat{r}'_{\theta,m}\big)$ with $\hat{r}'_{\theta,m}$ defined in \eqref{eq:rhat-2}. 


\section{Some simulation results}\label{sec:simulations}
\subsection{Linear Time Invariant system}
To illustrate the performance of the proposed estimators, we first consider a dynamical LTI system described by a model of the form \eqref{eq:system} 
where $\theta^\circ=\bbm -1.40 & 0.75 & 0.60 & -0.10\eem^\top \in \Re^4$ 
 and $x(t)=\bbm -y(t-1) & -y(t-2) & u(t-1) & u(t-2)\eem^\top\in \Re^4$ with the input $\left\{u(t)\right\}$ being generated as the realization of a zero-mean white Gaussian noise with unit variance. As to the  noise sequence $\left\{v(t)\right\}$, it is uniformly sampled from an interval of the form $\interval{-a}{a}$ with $a=0.2$. In these conditions, we consider an estimation horizon of length $N=200$ data points and compute the interval-valued parameter estimates described in \eqref{eq:interval-estimate}. The initial parameter set $\mathscr{I}(c_\theta(0),r_\theta(0))$ is selected such that $c_\theta(0)=0$ and $r_\theta(0) =\alpha_0\bm{1}_n $ with $\alpha_0=4$ and  $n=4$ here and $\bm{1}_n$ being a $n$-dimensional vector of ones. The reference RLS algorithm  \eqref{eq:update-theta}-\eqref{eq:update-P} is run with initial value $\theta(0)=0$, covariance matrix $P(0)=10^3 I_4$ and  forgetting factor $\lambda=0.99$.

\paragraph{Evaluations of the preliminary estimators}
Considering the truncated estimates \eqref{eq:realization-c}-\eqref{eq:rhat}, we start by recalling that, as established by Theorem \ref{thm:m-star}, there is a minimum value of the horizon $m$ beyond which boundedness of the estimate can be hoped for. With the experimental setting described above, a minimal  such value is empirically found to be about  $10$ for most realizations of the input-output data.

 Figure \ref{fig:Int} below presents the interval-valued parameter estimates for this example when applying the estimators described  \eqref{eq:realization-c}-\eqref{eq:rhat} for $m\in \left\{20,50,N\right\}$. Note that the case $m=N$ with $N$ being the entire estimation horizon generated an interval radius $\hat{r}_{\theta,m}$ such that $\hat{r}_{\theta,m}=r_\theta$ (See Eq. \eqref{eq:interval-estimate}). The results confirm that the estimate $\mathscr{I}(c_{\theta}(t),r_\theta(t))$ defined in \eqref{eq:interval-estimate} is tighter than  the truncated forms $\mathscr{I}(c_{\theta}(t),\hat{r}_{\theta,m}(t))$ for $m<N$. Moreover, the larger the truncation horizon $m$, the tighter $\mathscr{I}(c_{\theta}(t),\hat{r}_{\theta,m}(t))$.  
\begin{figure}[h]%
\centering
\psfrag{N20}{\scriptsize$m=20$}
\psfrag{N50}{\scriptsize$m=50$}
\psfrag{N}{\scriptsize$m=N$}
\psfrag{True}{\scriptsize True}
\psfrag{t}{\scriptsize Time}
\psfrag{theta}{\scriptsize Estimates}
\includegraphics[width=.8\columnwidth,height=.15\paperheight]{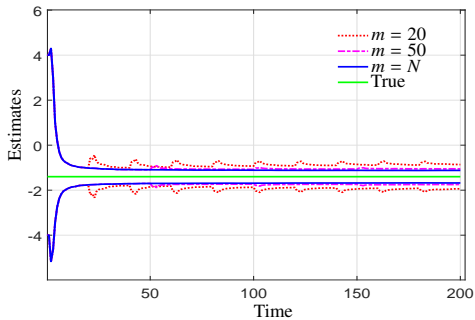}%
\caption{Interval-valued parameter estimates for the first entry of $\theta^\circ$ (averaged over $100$ independent runs). Truncated estimates $(c_\theta,\hat{r}_{\theta,m})$ with $\hat{r}_{\theta,m}$ as in \eqref{eq:rhat} for $m=20$ (dotted red), $m=50$ (dashed magenta), $m=N$ (solid blue) and the true (constant) parameter $\theta^\circ$ (solid green).  }%
\label{fig:Int}%
\end{figure}

\paragraph{Monotonically improved estimators}
To illustrate the benefit of the monotonic operators proposed in Section \ref{subsec:monotonic}, we apply them to the three previous estimators, that is, the estimates are now computed by using  $(c_\theta,\hat{r}_{\theta,m})$ as input in Eqs \eqref{eq:theta-inf}-\eqref{eq:theta-sup}, with $m\in \left\{10,N\right\}$. Again recall that $\hat{r}_{\theta,N}=r_\theta$. The associated results are plotted in Figure \ref{fig:Int-Monotonic}. As argued before, we can see that the obtained estimates  are  smoother and tighter compared to those of Figure \ref{fig:Int}.  Moreover, they effectively generate intervals with monotonically decreasing (nonincreasing) widths. 
\rev{For comparison purpose, we have also represented estimates\footnote{\rev{Note that no line is visible in the time interval $[0,20]$ because the first $20$ samples are used here to initialize the algorithm. }} obtained by the method described in \cite{Gutman94}. As it turns out, our estimator gives tighter estimates. We will see in the next paragraph that more tightness can be gained by using a smaller forgetting factor $\lambda$. }
\begin{figure}[h]%
\centering
\psfrag{10}{\scriptsize$m=10$}
\psfrag{N}{\scriptsize$m=N$}
\psfrag{True}{\scriptsize True}
\psfrag{Gutman}{\scriptsize Ref \cite{Gutman94}}
\psfrag{t}{\scriptsize Time}
\psfrag{theta}{\scriptsize Estimates $(\underline{\theta},\overline{\theta})$}
\includegraphics[width=.77\columnwidth,height=.15\paperheight]{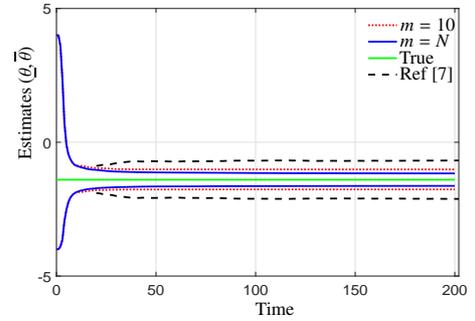}
\caption{Interval-valued  estimates (averaged over $100$ independent runs)   given by the estimator \eqref{eq:theta-inf}-\eqref{eq:theta-sup} with $\overline{\xi}=c_\theta+\hat{r}_{\theta,m}$, $\underline{\xi}=c_\theta-\hat{r}_{\theta,m}$ and $\hat{r}_{\theta,m}$ as in \eqref{eq:rhat} for $m=10$ (dotted red),  $m=N$ (solid blue) :  true (constant) parameter $\theta^\circ$ (solid green) and estimates given by the method of \cite{Gutman94} (dashed black).  }%
\label{fig:Int-Monotonic}%
\end{figure}

\paragraph{Influence of the RLS forgetting factor}
As can be intuitively guessed, the performance of the proposed interval-valued estimators depends on the properties of the RLS reference identifier which in turn are determined by the richness of the data and the user-defined parameters such as $P(0)$ and $\lambda$. In particular, it is interesting to study the impact of the forgetting factor $\lambda\in \interval[open left]{0}{1}$. In general, for point-valued estimation, such a parameter is selected, to be close to $1$ in order to smooth the trajectories of $\theta$ in \eqref{eq:update-theta}. In contrast,  the  recursive interval-valued estimator \eqref{eq:theta-inf}-\eqref{eq:theta-sup} tends to perform better when $\lambda$ is small.     To see this consider Figure \ref{fig:Error-vs-lambda} where we have plotted the final interval width $\overline{\theta}(N)-\underline{\theta}(N)$ achieved by the estimator \eqref{eq:theta-inf}-\eqref{eq:theta-sup}. Again only the estimates related to the first component of the parameter vector are represented. We consider the estimator  \eqref{eq:theta-inf}-\eqref{eq:theta-sup} with $\underline{\xi} =c_\theta- \hat{r}_{\theta,m}$ and $\overline{\xi} =c_\theta+\hat{r}_{\theta,m}$  as defined in \eqref{eq:interval-estimate} and \eqref{eq:rhat} for $m\in \left\{20,50,N\right\}$. The results are indeed averages over $100$ independent simulations. What this reveals is that the estimator's asymptotic performance depends on the forgetting factor in the sense that the width of the estimated interval is all the smaller as the forgetting factor $\lambda$ is small. This behavior can be explained by the fact that a small $\lambda$ in the RLS may cause the estimates $(c_\theta,\hat{r}_{\theta,m})$ to fluctuate substantially hence favoring the event that the associated interval jumps occasionally to a small value.  We can further observe that for small values of the forgetting factor (e.g., $\lambda\leq 0.6$ in Figure \ref{fig:Error-vs-lambda}),  all truncation orders $m$ tend to perform equally well. This suggests an important feature of the proposed estimation framework for practical implementation: provided the exciting input $\left\{u(t)\right\}$ is sufficiently rich and  $\lambda$ is then taken  small enough, the computational complexity of the estimators can be reduced to the minimum by selecting a small truncation horizon $m$.   
\begin{figure}[h!]%
\centering
\psfrag{N10}{\scriptsize$m=20$}
\psfrag{N4}{\scriptsize$m=50$}
\psfrag{N}{\scriptsize$m=N$}
\psfrag{lambda}{\footnotesize$\lambda$}
\psfrag{Error}{\scriptsize Error $\overline{\theta}(N)-\underline{\theta}(N)$}
\includegraphics[width=.78\columnwidth,height=.14\paperheight]{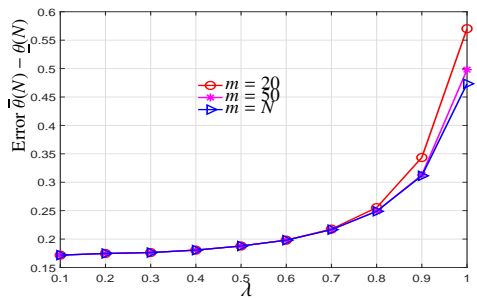}%
\caption{Widths $\overline{\theta}(N)-\underline{\theta}(N)$ (averaged over $100$ independent runs) of the estimated parameter intervals versus forgetting factor $\lambda$. Only the first components of the parametric  (vector-valued) error  are represented for truncated estimators of the form \eqref{eq:rhat} for $m=20$ (dashed red), $m=50$ (dashed magenta), $m=N$ (solid blue) and the true (constant) parameter $\theta^\circ$ (solid green).  }%
\label{fig:Error-vs-lambda}%
\end{figure}
\subsection{Linear Time Varying system}
We now consider a model of the form \eqref{eq:system} where the parameter vector $\theta^\circ$ is time-varying with dynamics defined as in \eqref{eq:dynamics-theta0} where it is assumed that $\delta(t)$ belongs to an interval given by $c_\delta(t)=0$ and  $r_\delta(t) =\bbm 0.10 & 0.05& 0.04 & 0.01\eem^\top $ for all $t$.  For the simulation, we generate a sequence $\left\{\delta(t)\right\}$ in this interval such that $\delta(t)=r_\delta(t)\sin(2\pi t/30)$.
\rev{
 The other settings remain the same as previously defined in the beginning of Section \ref{sec:simulations} except the forgetting factor which is now set to $0.1$ (recall that as discussed earlier, the estimate is tighter when $\lambda$ is small).   Consider applying the estimator \eqref{eq:theta-inf2}-\eqref{eq:theta-sup2} with inputs $\underline{\xi} =c_\theta- \hat{r}_{\theta,m}$ and $\overline{\xi} =c_\theta+\hat{r}_{\theta,m}$  as defined in \eqref{eq:interval-estimate} and \eqref{eq:rhat} for $m\in \left\{5,N\right\}$. The outcome of this experiment is depicted in Figure \ref{fig:Estimates-LTV}. 
For a value of $\lambda$ as small as $0.1$, the estimated interval appears to be very tight. Moreover, all values of the truncation horizon $m$ give almost the same performance in this case.  
}
\begin{figure}[h!]%
\centering
\psfrag{N5}{\scriptsize$m=5$}
\psfrag{N}{\scriptsize$m=N$}
\psfrag{True}{\scriptsize True}
\psfrag{t}{\scriptsize Time}
\psfrag{theta}{\scriptsize Estimates $(\underline{p},\overline{p})$}
\includegraphics[width=.8\columnwidth,height=.14\paperheight]{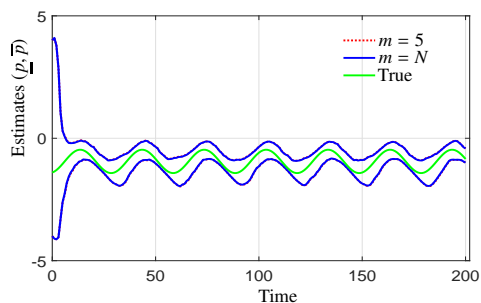}
\caption{Interval-valued parameter estimates (averaged over $100$ independent runs)  given by the estimator \eqref{eq:theta-inf2}-\eqref{eq:theta-sup2} on the time-varying example with  $\overline{\xi}=c'_\theta+\hat{r}'_{\theta,m}$, $\underline{\xi}=c'_\theta-\hat{r}'_{\theta,m}$ and $\hat{r}'_{\theta,m}$ as in \eqref{eq:rhat-2} for $m=5$ (dashed red) and $m=N$ (solid blue). 
The true time-varying parameter $\theta^\circ$ is in solid green.  }%
\label{fig:Estimates-LTV}%
\end{figure}
\vspace{-.2cm}
\section{Conclusion}\label{sec:conclusion}
\noindent In this paper, we have presented a recursive interval-valued estimation framework for the identification  of linearly parametrized models. The main idea of the method is to carefully bound the error generated by a certain reference adaptive algorithm, for example the recursive least squares. However the smallest interval-valued estimator we discussed  turns out to be computationally costly to implement in an online scenario. We therefore turn to an alternative family of (over)-estimators which exhibits a trade-off between the achievable performance of and the price to pay for it in computational load.  Two cases have been studied: one where the to-be-estimated parameter vector is constant and a more general situation where it is possibly time-varying. In the first case, we further show that the estimated interval size can  be made monotonically decreasing. In the second, this monotonic property cannot be systematically achieved (as this depends on the  change rate of the parameters)  but the width of the estimated interval can be made very small by an appropriate design of the reference point-value identifier. For example, we have observed in simulation that when the reference identifier is the RLS algorithm, the performance of the estimator improves if the forgetting factor is small.  \\
Future work may concern the extension of the proposed interval-valued estimation framework to systems whose models are nonlinear in the parameters.

\bibliographystyle{abbrv}

\end{document}